\documentclass[12pt]{article}

\usepackage[default]{jasa_harvard}    % for formatting citations in text
\usepackage{JASA_manu}

\usepackage{graphicx,subfigure,amsmath,latexsym,amssymb}
\usepackage{float,epsfig,multirow,rotating,times}
\usepackage{upgreek,wrapfig}
\usepackage{comment}

\newcommand{\bg}{\boldsymbol{g}}

\newcommand{\btheta}{\boldsymbol{\theta}}

\newcommand{\bSigma}{\boldsymbol{\Sigma}}

\newcommand{\bzero}{\boldsymbol{0}}

\newtheorem{theorem}{Theorem}

\newenvironment{proof}[1][Proof]{\textbf{#1.} }{\ \rule{0.5em}{0.5em}}

\renewcommand{\t}{\ensuremath{\theta}}

\numberwithin{equation}{section}
\numberwithin{algo}{section}
\numberwithin{table}{section}
\numberwithin{figure}{section}

\begin{document}

\normalsize

\title{\vspace{-0.8in}
A Note on the Misuse of the Variance Test in Meteorological Studies}
\author{
Arnab Hazra, Sourabh Bhattacharya, Sabyasachi Bhattacharya, Pabitra Banik
\thanks{
Indian Statistical
Institute, 203, B. T. Road, Kolkata 700108.
Corresponding e-mail: sourabh@isical.ac.in.}}
\date{\vspace{-0.5in}}
\maketitle%

\begin{abstract}

The erroneous assumption ``for all distributions for which the theoretical variance can be computed independently 
from parameters estimated by any method different from the method of moments" has been used in the case of fitting 
the gamma distribution to a rainfall data by Mooley (1973) which was followed by several researchers. 
We show that the asymptotic distribution of the test statistic is generally not even comparable to any 
central chi-square distribution. We also describe a method for checking the validity of the asymptotic distribution 
for a class of distributions.
\\[2mm]
\end{abstract}

\section{Introduction}
The variance ratio test statistic provides a measure of goodness-of-fit. In the spirit of the pioneering idea 
of Fisher (1925) as an index of dispersion, Cochran (1954) efficiently used and popularized this test,  
illustrating with examples in the case of small samples from the Poisson and the Binomial series. The test statistic 
was referred to as the central chi-square distribution with degrees of freedom one less than the sample size 
in both the cases, whether the parameters are specified or not, and a proof of this fact was given by 
Rao and Chakravarti (1956) for the Poisson series. For large sample size, a modified form of the test statistic 
was proposed by Fisher and Yates (1957) so that its asymptotic distribution corresponds to the standard normal density.
%\vskip1.mm
 
Perhaps due to ignorance related to the asymptotic theory of the variance ratio test statistic, Mooley (1973) committed 
 a misuse of this goodness-of-fit test in the case of fitting a gamma distribution to Asian summer monsoon rainfall 
 data, assuming that the test can be used for all distributions for which the theoretical variance can be 
 computed independently of parameters estimated by a method other than the method of moments. He also assumes 
 that the asymptotic distribution of the test statistic in such a situation would be a chi-square distribution 
 with degrees of freedom one less than the sample size. A significant number of other authors (Hargreaves (1975), 
 Sarker et al. (1982), Biswas et al. (1989), Goel and Singh (1999)) followed the similar misuse synergistic 
 with the work of Mooley (1973).
%\vskip1.mm

As far our knowledge is concerned, no potential work has yet explored the fact that the implementation of the variance 
ratio test by Mooley (1973) was incorrect. Several studies have been conducted on rainfall analysis and the 
best fit probability distribution function such as the gamma distribution function (Barger and Thom (1949), Mooley and Crutcher (1968), 
Sen and Eljadid (1999)), log-normal (Sharma and Singh (2010), Kwaku et al. (2007)), exponential (Duan et al. (1995), 
Burgueno et al. (2005), Todorovic and Woolhiser (1975)), Weibull (Duan et al. (1995), Burgueno et al. (2005)) 
distributions were identified under different situations. Our simulation study suggests that when the data sets 
have the proximity to any one of exponential, gamma, Weibull, log-normal, the usual asymptotic distribution of the 
test statistic is no longer even central chi-square and thus the variance ratio test can not be used under any 
of the circumstances. Below we provide a brief overview of the issues involved.

\subsection{Variance ratio test and its misuse}
\label{subsec:test_stat}

Suppose we want to fit the random sample  $X_1,\ldots, X_n$ to a distribution whose cumulative distribution function 
is given by $F$. We consider the following hypothesis testing problem -- 
$H_0$: the sample comes from the distribution $F$,
versus $H_1$ : the sample does not come from the distribution $F$. The test statistic proposed by Fisher (1925) 
and illustrated by Cochran (1954) is 
\begin{eqnarray}
\nonumber {\chi_\nu}^2 &=& \sum_{i=1}^{n}\frac{{(X_i-\overline{X})}^2}{\widehat{{\sigma_F}^2}}, 
\end{eqnarray}
where $\widehat{{\sigma_F}^2}$ is the estimate of the population variance computed independently of parameters 
estimated by a method other than the method of moments. %({\bf\large ekhane further clarification dorkar}). 
So, this method is not applicable to fitting distributions 
like normal. The test statistic was used in the case of Poisson and Binomial distributions by Fisher (1925) but no proper 
mathematical justification was provided. Cochran (1936) and Rao and Chakravarty (1956) calculates a form of the 
approximate expression of the mean and variance of the distribution of the test statistic and also provide justification 
of implementing this test for Poisson and Binomial distributions under the null hypothesis.
%\vskip1.mm
 
Mooley (1973) uses this test for fitting gamma distribution to the monsoon rainfall data and estimates the 
unknown parameters using maximum likelihood estimation. Thus, the estimate of the population variance and the 
sample variance do not coincide. Mooley (1973) states that the method works well for any probability model 
satisfying such criterion. In Section \ref{sec:exposition} we provide theoretical justification why the statement 
is incorrect, demonstrating the issue with the exponential distribution; we also conduct a simulation
study, justifying the same with several other distributions -- gamma, log-normal and Weibull. 
%which also satisfies the criterion. 
In Section \ref{sec:check_validity} we identify a large class of distributions where we can simply 
check whether or not the usual 
asymptotics as in the cases of Poisson and Binomial are valid.
In particular, we discuss and derive the asymptotics of the variance ratio 
test for a large class of distributions with finite fourth moment when the population variance can be written 
as a differentiable function of the population mean.
%\vskip2.mm

\section{Analytical and empirical exposition of the misuse of the variance ratio test}
\label{sec:exposition}

Suppose that $X_1,X_2,\ldots,X_n$ is a random sample of size $n$ from the exponential distribution with 
mean $\lambda$, where the density is given by 
\begin{eqnarray}
      \nonumber f(x) &=& \frac{1}{\lambda} e^{-\frac{x}{\lambda}}.
\end{eqnarray}
In the above, the mean $\lambda$ is unknown, which we assume to be estimated from the sample using the 
maximum likelihood estimation (MLE) 
procedure. 

With the above set-up, Theorem \ref{theorem:theorem1} shows that the asymptotic distribution
of the variance ratio test statistic is not $\chi^2_{n-1}$.

\begin{theorem}
\label{theorem:theorem1}
The asymptotic distribution of the variance ratio test statistic, under the null hypothesis that a random 
sample of size $n$ comes from a one-parameter exponential distribution, is not comparable with central 
$\chi^2 $ distribution with $n-1$ degrees of freedom.
\end{theorem}
\begin{proof}
The estimate of the unknown parameter $\lambda$ is 
$\lambda_{MLE}=\frac{X_1+ X_2+\cdots+ X_n}{n}=\overline{X}$, that is, the sample mean. 
The population variance is $\lambda^2$ and thus, the MLE of the population variance is 
${\widehat{\lambda^2}}_{MLE}={({\widehat{\lambda}}_{MLE})}^2={\overline{X}}^2$. Hence, the test statistic 
in our case is given by
\begin{eqnarray}
\nonumber D &=& \sum_{i=1}^{n}\frac{{(X_i-\overline{X})}^2}{\overline{X}^2}. 
\end{eqnarray}

%We follow the procedure of Rao and Chakravarti (1956) to provide justification of the asymptotic properties 
%of the variance ratio test in the case of Poisson distribution. Some necessary modifications are done as we 
%are dealing with a continuous distribution here.

For the remaining part of the proof we follow Rao and Chakravarti (1956) who provide 
justification of the asymptotic properties of the variance ratio test in the case of the Poisson distribution,
but make necesary modifications
to accommodate our case of continuous distribution. 

We note that the sample total given by $T=X_1+ X_2+\cdots+ X_n$ is sufficient for $\lambda$. 
Here $T$ follows the gamma distribution with shape parameter $n$ and scale parameter $\lambda$; the density 
is given by
\begin{eqnarray}
\nonumber f_T(t) &=& \frac{1}{\lambda^n \Gamma(n)}e^{-\frac{t}{\lambda}}t^{n-1}.
\end{eqnarray}

The conditional density of $X_1,X_2,\ldots,X_n$ given $T=t$ is given by 
\begin{eqnarray}
\nonumber f_{X_1,X_2,\ldots,X_n |T=t}(x_1,x_2,\ldots,x_n )=\frac{\Gamma(n)}{t^{n-1}}. 
\end{eqnarray}

Now, $E(X_i|T)= \overline{X}=T/n$ and thus we can express the variance ratio test statistic in the form
\begin{eqnarray}
\nonumber D &=& \sum_{i=1}^{n}\frac{{(X_i-E(X_i|T))}^2}{E(X_i|T)^2}. 
\end{eqnarray}

%To investigate the nature of the asymptotic distribution of $D$ we obtain the first two moments of the 
%test statistic subject to the condition that $T$ and $n$ are fixed.
Now, by the definition of conditional expectation, we have, 
for any measurable function $\phi(x_1,x_2,\ldots,x_n)$: 
%is connected with its total 
%expectation by the following relation, where $\lambda$ is the common mean of the individual exponentials
\begin{eqnarray}
\nonumber \int_{0}^{\infty}E(\phi|T=t)f_T(t)dt &=& E(\phi),
\end{eqnarray}
which, in our case, translates into
\begin{eqnarray}
\nonumber \int_{0}^{\infty}E(\phi|T=t)e^{-\frac{t}{\lambda}}t^{n-1}dt &=& E(\phi){\lambda^n}\Gamma(n). 
\end{eqnarray}

Therefore, knowing the total expectation $E(\phi)$, the conditional expectation $E(\phi|T=t)$ can be easily obtained. 
Let us consider the statistic $S^2=\sum_{i=1}^n(X_i-\overline{X})^2=\sum_{i=1}^{n}{X_i}^2-n\overline{X}^2$
whose moments are known functions of $\lambda$. Using the above definition of conditional expectation
we derive the conditional moments of $\phi(x_1,x_2\ldots,x_n)=S^2$ as follows. 

Since $E(S^2)=(n-1)\lambda^2$, we have
\begin{eqnarray}
\nonumber \int_{0}^{\infty}E(S^2|T=t)e^{-\frac{t}{\lambda}}t^{n-1}dt &=& (n-1){\lambda^{n+2}}\Gamma(n). 
\end{eqnarray}

Now, we can write $\lambda^{n+2}=\int_{0}^{\infty}\frac{1}{\Gamma(n)}e^{-\frac{t}{\lambda}}t^{n+1}dt$, 
and thus it follows that
 \begin{eqnarray}
 \nonumber \int_{0}^{\infty}E(S^2|T=t)e^{-\frac{t}{\lambda}}t^{n-1}dt 
 &=& \int_{0}^{\infty}\frac{(n-1)\Gamma(n)}{\Gamma(n+2)}e^{-\frac{t}{\lambda}}t^{n+1}dt.
 \end{eqnarray}

We know that if $\int_{0}^{\infty}f_1(x)e^{-ax}dx=\int_{0}^{\infty}f_2(x)e^{-ax}dx$ where $f_1(x),f_2(x)$ 
and both are continuous, $a$ is some positive constant, then, $f_1=f_2$  by the uniqueness of the 
Laplace transform. As a consequence,
\begin{eqnarray}
\nonumber E(S^2|T=t) &=& \frac{(n-1)}{n(n+1)}t^2,
\end{eqnarray}  
and hence
\begin{eqnarray}
E(D|T=t) &=& E(S^2\frac{n^2}{T^2}|T=t)\nonumber \\
&=& \frac{n^2}{t^2}\frac{(n-1)}{n(n+1)}t^2\nonumber \\
&=& \frac{(n-1)n}{n+1}\nonumber \\
&\approx& n-1.\nonumber
\end{eqnarray}

Similarly, we obtain
\begin{eqnarray}
\nonumber E(S^4) &=& \frac{(n-1)(n^2+7n-6)}{n}\lambda^4
\end{eqnarray}
and
\begin{eqnarray}
\nonumber E(S^4|T=t) &=& \frac{(n-1)(n^2+7n-6)\Gamma(n)}{n\Gamma(n+4)}t^4.
\end{eqnarray}
Thus,
\begin{eqnarray}
\nonumber E(D^2|T=t) \\
\nonumber &=& E({S^4}\frac{n^4}{T^4}|T=t) \\
\nonumber &=& \frac{n^4}{t^4}\frac{(n-1)(n^2+7n-6)\Gamma(n)}{n\Gamma(n+4)}t^4 \\
\nonumber &=& \frac{{n^2}(n-1)(n^2+7n-6)}{(n+3)(n+2)(n+1)},
\end{eqnarray}
and so,
\begin{eqnarray}
\nonumber Var(D|T=t) \\
\nonumber &=& E(D^2|T=t)-{(E(D|T=t))}^2 \\
\nonumber &=& 4(n-1)\frac{1}{{(1+\frac{1}{n})^2}(1+\frac{2}{n})(1+\frac{3}{n})}
\cr \nonumber &\approx& 4(n-1).
\end{eqnarray}

Since $Var(D|T=t)$ is independent of $t$, it follows that $Var(D)=Var(D|T=t)$, which does not conform 
with the variance of the central chi-square distribution with $(n-1)$ degrees of freedom which is $2(n-1)$. 
This proves the theorem.

\end{proof}

\subsection{Simulation study to demonstrate the effect of the erroneous assumption}
\label{subsec:simstudy}

%We have assumed some approximations in the calculation and thus we perform simulation studies 
%also to demonstrate that the approximations work fairly. 

To demonstrate the effect of the erroneous assumption of $\chi^2_{n-1}$ as the asymptotic distribution
of $D$, we consider a simulation study pertaining to the cases of exponential, gamma, 
log-normal and Weibull. We calculate the values of the empirical mean and empirical 
variance for different values of the parameters for different null distributions, based on 
10,000 simulated samples in each case. 

The results are presented in Table \ref{table:table1}.
Correct usage of Cochran's variance ratio test should yield the mean and the variance close to 100 and 200 respectively 
in case (a) and 200 and 400 respectively in case (b). 
However, the results in Table \ref{table:table1} are far from the aforementioned values, clearly 
pointing towards incorrect implementation of the test. 
    
    \begin{flushleft}
       
    \begin{table} [h!]
    \label{table:table1}
             \begin{flushleft}

             \caption{Table of means and variances of the variance ratio test statistic: (a) sample size 100, (b) sample size 200 for different distributions}
             \vspace{1em}
             \vspace{1em}

             \begin{tabular}{cccccccc}
                & & \multicolumn{6}{c}{Parameter values} \\ \hline
                          &   &  & 1 & 5 & 10 & 15 & 20 \\ \hline
  Exponential     & (a)	& mean& 97.68	& 97.79	& 98.15 &	97.92 &	97.92  \\
  (mean=parameter)     &  & variance	&	364.64	& 360.17	& 373.09 &	370.65 &	362.63  \\
         & (b)	& mean	&	197.37	& 197.96	& 198.07 & 197.99	& 198.48  \\
              &  & variance	&	756.86	& 754.43	& 765.76 &	765.13	& 786.41  \\
     Gamma (scale=2,                 & (a)	& mean	&	98.95	& 99.78		& 99.97	& 99.96	& 99.93  \\
shape=parameter)        &   & variance	&	220.38	& 48.69		& 24.45	& 16.73	& 12.22  \\
             &   (b) &	mean	&	199.04	& 199.95	& 199.94 &	199.82 &	199.95  \\
                        &  & variance	&	450.37	& 100.65	& 51.21	& 34.52	& 25.42  \\
     Gamma (shape=2,           &  (a)	& mean	&	99.56	& 99.44	& 99.65	& 99.39	& 99.37  \\
  scale=parameter)      &  & variance	&	116.67	& 112.15 &	120.43 &	117.56 &	113.22  \\
             & (b)	& mean	&	199.48	& 199.37 &	199.26	& 199.39 &	199.30  \\
                         &     & variance	&	240.53	& 245.83 &	242.56 &	243.29	& 241.27  \\
   Lognormal     &  (a) & mean		&	89.29	& 98.59	& 98.95	& 98.93	& 99.00  \\
   (location=parameter,       &      & variance	&	5692.70	& 52.53	& 11.67	& 4.96	& 2.75  \\
    scale=2)     &      (b)	& mean	&	191.55	& 198.51 &	198.85 &	198.93 &	198.96  \\
                       &      & variance	&	19131.99 & 108.18 &	23.68 &	10.47 &	5.78  \\
   Lognormal    &     	(a)	& mean	&	94.36	& 84.63	& 73.93	& 80.47	& 62.68  \\
  (scale=parameter,       &   & variance	&	1608.40	& 70834.68 & 32260.12 & 700636.96 & 79507.59  \\
   location=1)     &  (b)	& mean	&	194.34	& 181.99	& 176.57 &	167.81 &	168.77  \\
                       &  & variance	&	3867.04	& 71164.76 & 513858.83	& 1669349.96 & 1529708.93  \\
    Weibull (shape=2,      &  (a)	& mean	&	100.08	& 100.03 &	100.06 &	100.06 &	100.05  \\
    scale=parameter)       &     & variance	&	3.27	& 3.30	& 3.17	& 3.42	& 3.31  \\
           &    (b) &	mean	&	200.06	& 200.08 &	200.06	& 200.05 &	200.02  \\
         &   & variance	&	6.18	& 6.27	& 6.34	& 6.14	& 6.37  \\
     Weibull (scale=1,      &	(a) &	mean	&	42.30	& 98.70	& 100.07 &	100.18	& 100.17  \\
     shape=parameter/5)    &  & variance	&	144990.96 & 130.44 & 3.22 &	5.60 &	16.31  \\
       &  	(b) &	mean	&	108.65	& 198.97 &	200.07 &	200.20	& 200.19  \\
                      &   & variance	&	1667054.35 &	292.83 &	6.18 &	10.86 &	32.22  \\      

    \hline 
              \end{tabular}
                       %\label{TAB:1}
                       \end{flushleft}
                       \end{table} 
                       \end{flushleft}  
%\vskip1.mm
%    If Cochran’s variance ratio test would work properly, the mean and the variance should be close enough to 100 and 200 respectively in case (a) and 200 and 400 respectively in case (b) but Table 1 demonstrates that the implementation of this test in such situations is incorrect.
%\vskip2.mm    

\section{Checking the validity of the $\chi^2$ assumption for the asymptotic distribution of the variance
ratio test statistic}
\label{sec:check_validity}

Theorem \ref{theorem:theorem2} below provides a way to check the validity of the 
$\chi^2_{n-1}$ assumption for the asymptotic distribution of $D$.
 
 \begin{theorem}
 \label{theorem:theorem2}
 If a random sample of size $n$ comes from a population with finite fourth moment 
 where the population variance is a differentiable function $f$ of the population mean under 
 the null hypothesis, then under the condition (which we refer to as the ``condition of approximate equality")
\begin{eqnarray}
\nonumber \frac{1}{f(\mu)^4}\left(\sigma^6{(f'(\mu))^2}-2\mu\sigma^2\mu_3f'(\mu)+{f(\mu)^2}(\mu_4-\sigma^4)\right) &\approx& 2,
\end{eqnarray}
where $\mu, \sigma^2, \mu_3, \mu_4$ are the mean, 2nd, 3rd and 4th central moments of the population respectively, 
the variance ratio test statistic is asymptotically central $\chi^2$ with $n-1$ degrees of freedom. If a function like $f$ exists 
and the ``condition of approximate equality" fails, then the variance ratio test statistic is not asymptotically $\chi^2_{n-1}$.
\end{theorem}
\begin{proof}
Suppose that $X_1,X_2,\ldots,X_n$ is a random sample from a population where the sufficient condition on moment
existence and the existence of a differentiable function $f$ are satisfied under the null hypothesis.
    
Applying the bivariate central limit theorem (CLT) in the context of sample moments, we obtain the joint 
asymptotic distribution of sample mean $\overline{X}_n$ and sample variance $S^2_n$ as
        \begin{eqnarray}
        \sqrt{n}\left[\begin{pmatrix}X_n\\{S_n}^2\end{pmatrix}-\begin{pmatrix}\mu\\\sigma^2\end{pmatrix}\right] 
       &\rightarrow & N_2\left(\begin{pmatrix}0\\0\end{pmatrix},\begin{pmatrix}\sigma^2 \hspace{8mm} \mu_3\\
       \mu_3  \hspace{4mm} \mu_4-{\sigma}^4\end{pmatrix}\right)\quad \mbox{in distribution}, 
       \label{eq:clt}
        \end{eqnarray}
where $\mu$ is the population mean, $\sigma^2$ is the population variance, $\mu_3$ and $\mu_4$ are the third and the 
fourth central moments of the population, respectively.
    
In the case of asymptotic normality of smooth functions of sample moments, it was shown by Cramer (1946) that for 
a mapping $\bg :{\Re^d}\rightarrow {\Re^k}$ such that $\bg'(\bf x)$, the derivative of $\bg(\bf x)$ at the point $\bf x$, 
is continuous in a neighborhood of 
$\btheta \epsilon \Re^d$, if $\bf T_n$ is a sequence of $d$-dimensional random vectors such that 
$\sqrt{n}(\bf T_n-\btheta)\rightarrow N_d(\bf 0,\bSigma)$
where $\bSigma$ is a $d\times d$ covariance matrix, then
    
    \begin{eqnarray}
    \sqrt{n}(\bg(\bf T_n)-\bg(\btheta))\rightarrow N_k(\bzero,\bg'(\btheta)\bSigma {\bg'(\btheta)}^T) 
    \quad\mbox{in distribution}. 
       \label{eq:delta_method}
    \end{eqnarray}
    
In our case, the population mean is estimated by the sample mean, and since $\sigma^2=f(\mu)$, 
the population variance is estimated by $f(\overline{X}_n)$ which is neither equal nor proportional to $S^2_n$ 
(otherwise $\sigma^2$ can not be written as a function of $\mu$ only). 
Hence %${\chi_\nu}^2
$D=\sum_{i=1}^{n}\frac{{(X_i-\overline{X}_n)}^2}{f(\overline{X}_n)}=
n\frac{{S_n}^2}{f(\overline{X}_n)}$ can be used as a test statistic. 
Then using the delta method (\ref{eq:delta_method}) associated with (\ref{eq:clt}), we have
    \begin{eqnarray}
    \sqrt{n}\left(\frac{S^2_n}{f(\overline{X}_n)}-\frac{\sigma^2}{f(\mu)}\right) &\rightarrow&
    N\left(0,\frac{1}{f(\mu)^4}\left(\sigma^6{(f'(\mu))^2}-2\mu\sigma^2\mu_3f'(\mu)+{f(\mu)^2}(\mu_4-\sigma^4)\right)\right)
    \nonumber\\
    \label{eq:asymp}
     \end{eqnarray} in distribution.

%The above asymptotic distribution 
%We shall now provide a brief justification behind the fact that for large values of $n$, $N(n,2n)\approx{\chi^2}_{n-1}$.
%Consider an independent and identical ($iid$) sequence of random variables  $Y_1,\ldots,Y_{n-1},\cdots$ which follow 
%$N(0,1)$. Then, ${{Y_1}^2}+\cdots+{{Y_{n-1}}^2}$ follows ${\chi^2}_{n-1}$. Now, by CLT, 
%$\frac{{{Y_1}^2}+\cdots+{{Y_{n-1}}^2}-(n-1)}{\sqrt{2(n-1)}}\rightarrow N(0,1)$ in distribution. 
%Thus, ${{Y_1}^2}+\cdots+{{Y_{n-1}}^2}$ has asymptotic distribution $N((n-1),2(n-1))$. 
%Now, for large enough $n$, $N((n-1),2(n-1))\approx N(n,2n)$. 
%%We can also justify using the fact that the ratio of the densities of ${\chi^2}_{n-1}$
%%and $N(n,2n)$ converges to $1$ as $n$ tends to infinity for any value in their common support. 
%Thus, the test statistic can be referred to the central chi-square table with $n-1$ degrees of freedom.
    
Now, if
 \begin{eqnarray}
 \nonumber \frac{1}{f(\mu)^4}(\sigma^6{(f'(\mu))^2}-2\mu\sigma^2\mu_3f'(\mu)+{f(\mu)^2}(\mu_4-\sigma^4)) &\approx& \alpha,
 \end{eqnarray}
 where $\alpha=2$, then the asymptotic distribution (\ref{eq:asymp}) is $N(n,2n)$. Since 
 $N(n,2n)\approx\chi^2_{n-1}$, in this case the variance test statistic $D$ 
 is asymptotically distributed as $\chi^2_{n-1}$.

On the other hand, if $\alpha$ is significantly different from 2, then $E(D)=n$ but 
$Var(D)\neq 2(n-1)$, even asymptotically. Hence, the asymptotic distribution of $D$
can not be central $\chi^2_{n-1}$ in this case.

%can't be approximated by $n\alpha$ and so by $2n$ while variance of the central chi-square distribution with degrees of freedom $n-1$, is equal to $2(n-1)$ which is comparable with $n\alpha$ for large $n$ only in case $\alpha\approx2$. Hence the test statistic can’t be referred to the central $\chi^2$ table with $n-1$ degrees of freedom.
\end{proof}

The following examples can be viewed as corollaries to Theorem \ref{theorem:theorem2}.

\subsection{Illustrative examples}
\label{subsec:examples}

\subsubsection{Poisson case}
\label{subsubsec:poisson}    
In the case of Poisson distribution with parameter $\lambda$, the delta method with $g(x,y)=\frac{y}{f(x)}$ 
and $f(x)=x$ yields 
\begin{eqnarray}
\nonumber \frac{1}{f(\mu)^4}\left(\sigma^6{(f'(\mu))^2}-2\mu\sigma^2\mu_3f'(\mu)+{f(\mu)^2}(\mu_4-\sigma^4)\right) &=& 2.
\end{eqnarray}
 Hence, 
$D=\sum_{i=1}^{n}\frac{{(X_i-\overline{X}_n)}^2}{\overline{X}_n}
=n\frac{S^2_n}{\overline{X}_n}$
has asymptotic distribution $N(n,2n)\approx \chi^2_{n-1}$.

\subsubsection{Binomial case}
\label{subsubsec:binomial}
In the case of Binomial distribution with size $M$ and probability $p$, applying the delta method 
with $g(x,y)=\frac{y}{f(x)}$ and $f(x)=\frac{x(M-x)}{M}$, we obtain 
\begin{eqnarray}
\nonumber \frac{1}{f(\mu)^4}\left(\sigma^6{(f'(\mu))^2}-2\mu\sigma^2\mu_3f'(\mu)+{f(\mu)^2}(\mu_4-\sigma^4)\right) 
&=& 2+\frac{1}{M}.
\end{eqnarray}
Since, for large enough $M$, $\frac{1}{M}\approx 0$,  
$D=\sum_{i=1}^{n}\frac{{(X_i-\overline{X}_n)}^2}
{\overline{X}_n(M-\overline{X}_n)}=n\frac{S^2_n}{\overline{X}_n(M-\overline{X}_n)}$ 
has $N(n,2n)\approx \chi^2_{n-1}$ as the asymptotic distribution.
    
\subsubsection{Exponential case}
\label{subsubsec:exponential}
In the case of exponential distribution with mean $\lambda$, let $g(x,y)=\frac{y}{f(x)}$ and $f(x)=x^2$. 
The delta method then yields
\begin{eqnarray}
\nonumber \frac{1}{f(\mu)^4}\left(\sigma^6{(f'(\mu))^2}-2\mu\sigma^2\mu_3f'(\mu)+{f(\mu)^2}(\mu_4-\sigma^4)\right) &=& 4.
\end{eqnarray}
Hence, $D=\sum_{i=1}^{n}\frac{{(X_i-\overline{X}_n)}^2}{{\overline{X}_n}^2}
    =n\frac{S^2_n}{{\overline{X}_n}^2}$ has the asymptotic distribution $N(n,4n)$, 
which can not be approximated by $\chi^2_{n-1}$. So, the test can not be used for exponential distributions. 
Note that this method of validation provides a straightforward way of proving Theorem \ref{theorem:theorem1}.
    
\section{Relevance of the study in rainfall data}    
\label{sec:rainfall}

Using data from 39 well-distributed and long-record stations over a relevant study region, and implementing the
$\chi^2$ goodness-of-fit test, the Kolmogorov-Smirnov test and the variance ratio test,
Mooley (1973) found the two-parameter gamma distribution to be the most suitable probability model 
among the Pearsonian models that show good fit to monthly rainfall in the Asian summer monsoon.
%using data from 39 well-distributed and long-record stations over the area implementing Chi-square goodness-of-fit test, 
%Kolmogorov-Smirnov test and variance ratio test. 
After implementing the variance ratio test in the context of weekly rainfall total, 
Hargreaves (1975) obtained the two-parameter incomplete gamma distribution suitable for the modeling purpose.
%to the weekly rainfall total 
%%(that is, fitted the two-parameter gamma distribution to the non-zero observations) 
%implementing the variance ratio test, considering week as an independent unit, and the lowest amounts expected at 
%different risks were calculated and considered 75\% probability as an acceptable risk value for most 
%conditions on a monthly basis. 
Sarker et al. (1982) computed the lowest amount of rainfall in 
the dry farming tract of north-west and south-west India
at different probability levels 
by fitting the same probability model, which was obtained by implementing the same variance ratio test.
On the basis of the same model they also considered 50\% probabilistic rainfall as dependable precipitation 
on a weekly basis.
%to week by week total non-zero rainfall of different stations in dry farming tract of north-west 
%and south-west India and also considered 50\% probabilistic rainfall as dependable precipitation 
%on a weekly basis. 
Biswas and Khambete (1989) computed the lowest amount of rainfall at different probability levels by 
fitting the same model, which was obtained by implementing the same variance ratio test 
on a data regarding week by week total rainfall of 82 stations in dry farming tract of Tamilnadu state of 
south-east India. Goel and Singh (1999) fitted the weekly rainfall data of Soan catchment in 
sub-humid area of Shivalik region of northern India to the same model, which they obtained by implementing the same test.

\section{Effects of the erroneous assumption by Mooley (1973) on inference: illustrations with simulated and real data}
\label{sec:investigate_effects}

Among the 39 Rain gage stations considered in Mooley (1973), the null hypothesis that the monthly rainfall 
series follow gamma distribution, was rejected for three cases. In particular, the null hypothesis associated with
the September rainfall of Allahabad, India, 
and July rainfall of Zi-Ka-Wei, China, were both rejected at level 0.05, using the variance ratio test. 
The test statistic in the case of June rainfall of Nagpur, India, was found to be significant at level 0.01. 
In case of the $\chi^2$ test, the null hypothesis was accepted at level 0.05 for all the considered cases. 

For adequate investigation of the above results obtained by Mooley, the actual data set used in Mooley (1973) 
is necessary. But unfortunately the data set is unavailable.
As a result, we are compelled to conduct further simulation studies to demonstrate that Mooley's implementation can
lead to rejection of the correct null hypothesis and acceptance of the false null hypothesis with high probability. 
However, in Section \ref{subsec:realdata} we also investigate the effects of Mooley's erroneous assumption 
using a real data set obtained from an independent source.

\subsection{Simulation based illustration of false rejection and false acceptance of the null hypothesis 
using Mooley's implemetation}

\subsubsection{First simulation study: false rejection of $H_0$}
First we draw 100,000 samples of size 100 from the gamma distribution with scale parameter $\lambda= 2$ and shape parameter 
$\alpha= 0.5$; the histogram of the observed test statistic is presented in Figure \ref{fig:fig1}. 
Now, according to the claim of Mooley (1973), 
the test statistic should be distributed as $\chi^2$ with degrees of freedom $100-1=99$. We draw the cut-offs as the vertical 
lines for a goodness of fit test of level $0.05$. As we draw samples from the null hypothesis, the expected number of 
rejections should be 5000. But here we see that the number of rejections is $13,214$, which is far above than the expected
number of rejections. Thus, this experiment demonstrates that there is a high chance of rejection of the null hypothesis even if the 
sample actually arises from the distribution under $H_0$.

\begin{figure}
\centering
\includegraphics[width=4in]{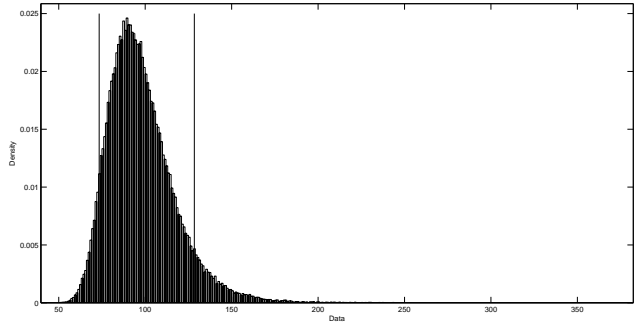}

\caption{Histogram of the null distribution of the test statistic; the vertical lines indicate the cut-off 
levels according to Mooley (1973).}
\label{fig:fig1}
\end{figure}

\subsubsection{Second simulation study: false acceptance of $H_0$}
We conduct another simulation study where we simulate 100,000 samples of size 30 from a mixture of three gamma distributions 
with equal weight (that is, each mixture component has mixing probability $\frac{1}{3}$). 
For the three gamma components, the parameters were chosen in such a way that the modes of the components are $1$, $5$ and $9$ 
respectively, while the variance under each component is specified to be $1$. 
The true mixture density, depicted in Figure \ref{fig:fig2}, is clearly significantly different from any single gamma distribution.

According to the claim of Mooley (1973), 
the number of cases of the rejection of the null hypothesis should be large enough, much more than 5\%, that is, 5000 cases. 
But in our simulated example only 3470 cases were rejected, even much less than the expected number of rejections 
under the null (the cases lying outside the cut-offs are shown in Figure \ref{fig:fig3}). This experiment thus demonstrates that
this test may often lead to false acceptance of the null hypothesis that the data is distributed as gamma while in reality 
the actual distribution is very far from gamma.

\begin{figure}
\centering
\includegraphics[width=4in]{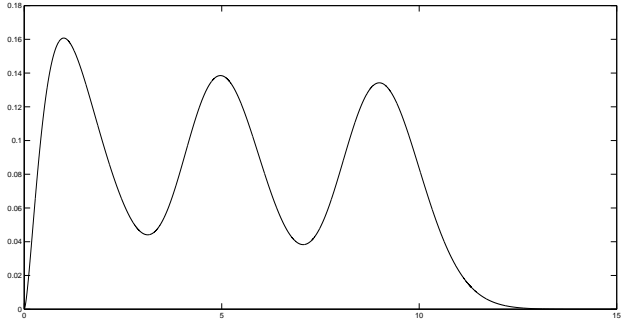}

\caption{Density of the mixture of three gamma distributions: significantly different from gamma.}
\label{fig:fig2}
\end{figure}

\begin{figure}
\centering
\includegraphics[width=4in]{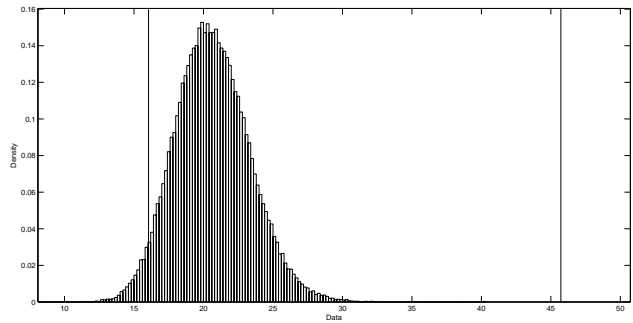}

\caption{Histogram of the null distribution of the test statistic; the vertical lines indicate the cut-off levels 
according to Mooley (1973).}
\label{fig:fig3}
\end{figure}

Next, we illustrate the issue of false rejection and false acceptance of the null hypothesis under Mooley's implementation
with a real data set.

\subsection{Illustration with June-September rainfall of India}
\label{subsec:realdata}

We obtain the dataset of All India Seasonal Rainfall Series (1901-2009) from the website of 
India Meteorological Department (http://www.imd.gov.in/section/nhac/dynamic/data.htm). In Figure \ref{fig:fig4}, we present 
the histogram of the observed dataset to which we fit a gamma distribution. For the $\chi^2$ goodness-of-fit test 
the $P$-value turns out to be $0.0167$, that is, for a test of level $0.05$, we reject the null hypothesis that 
the data is distributed as gamma. From Figure \ref{fig:fig4}, it is also evident that the fit is not ``good". Now, using the 
variance test in this set-up, the MLEs of the shape and scale parameters are $9.8663$ and $91.0873$ respectively, and 
the observed variance test statistic is $107.2916$. Assuming that the claim of Mooley (1973) about the asymptotic null 
distribution of the variance test statistic is true, the $P$-value turned out to be 
$$P(\left\vert Z\right\vert >|\sqrt{2 \cdot 107.2916}-\sqrt{217}|)=0.9344,$$
where $Z\sim N(0,1)$, leading to acceptance of the gamma distribution. However, poor fit exhibited by 
Figure \ref{fig:fig4}, rejection of the gamma distribution by the formal
$\chi^2$ test, and wisdom gained from our analytical and simulation based investigations regarding Mooley's implemntation
strongly suggests that this variance test wrongly accepts the false null hypothesis. 
%However, %based on the true distribution of the test statistic 
%the empirical $P$-value obtained from $100,000$ simulations assuming the MLEs to be the true values of the parameters, 
%turned out to be $0.1131$. This signifies the fact that the method proposed by Mooley (1973) may not be valid when the 
%observations arise from a gamma distribution and the parameters are estimated from the data using MLE.

\begin{figure}
\centering
\includegraphics[width=4in]{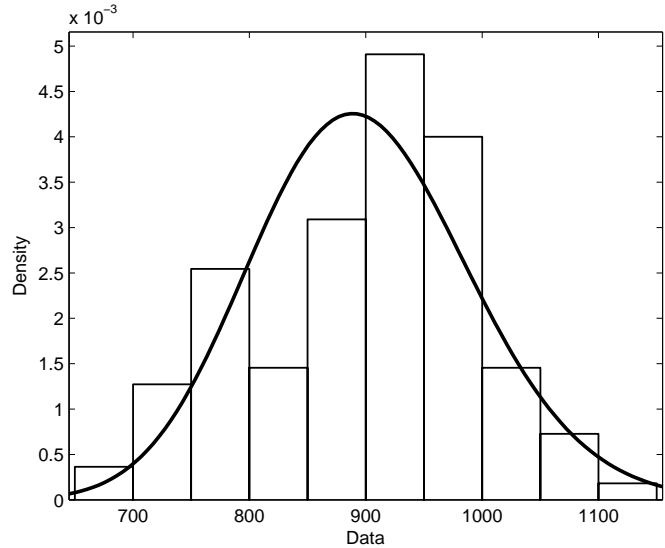} 

\caption{Histogram of the observed dataset and the fitted gamma density.}
\label{fig:fig4}
\end{figure}

\section{Discussion and Conclusion}
\label{sec:discussion}
The increased power of the variance test over Pearson's chi-square goodness-of-fit test was strikingly 
shown in some sampling experiments conducted by Berkson (1940), in a situation where the data followed a 
Binomial distribution. Berkson (1938) presented some data to illustrate the cases where the variance ratio 
test was significant but Pearson's chi-square goodness-of-fit test was not, in the contexts where 
the data followed a Poisson or a Binomial distribution.

However, when the underlying data follow a two-parameter gamma distribution, the asymptotic distribution 
of the variance ratio test statistic is largely dependent on the shape parameter and, as a consequence, the assumption 
that the test statistic asymptotically follows a central chi-square distribution, is erroneous 
and leads to misuse of the variance ratio test. 
Mooley (1973) seems to be the first to commit this misuse 
and a significant number of other authors followed the same path leading to misuse. 
Indeed, as we have shown in this article, for probability distributions like exponential, log-normal, Weibull, etc.,
which are frequently used in modelling rainfall data, the asymptotic distribution of the variance ratio test 
statistic is not commensurate with the chi-square distribution with degrees of freedom $n-1$.  
Hence, the test should be used very cautiously, particularly by meteorologists and other scientists.
 
To aid the meteorologists and the other practising scientists, in this article we have provided 
simple ways to check the validity of the variance ratio test for a
large class of distributions satisfying a few mild conditions.
%, it can be checked 
%whether or not it would be correct to use this test or not. 
In fact, as a necessary condition for applicability of the test, 
first it should be checked whether the limiting mean and variance are comparable 
with $n-1$ and $2(n-1)$ respectively.

If the variance ratio test is not applicable, it is better to use chi-square goodness-of-fit test 
in spite of having less power and loss of information by clubbing the data into different classes. 
At least it is theoretically correct and can be used in the case of fitting a mixture of zero rainfall and 
non-zero rainfall data. In the case of fitting non-zero rainfall data, it is more appropriate 
to use the Kolmogorov-Smirnov test than chi-square goodness-of-fit test in cases where the parameters 
under the null hypothesis are fully specified.
    
    \section*{References}
    \vskip2.mm
    \noindent Barger, Gerald L., and Thom, Herbert C. S., 1949: Evaluation of Drought Hazard, Agronomy Journal, Vol. 41, No. 11, Geneva, N.Y., NOV., pp. 519-526.
    \vskip3.mm
     \noindent Berkson, J., 1938: Some difficulties of interpretation encountered in the application of the chi-square test. J. Amer. Statist. Assoc., 33, 526-536
    \vskip3.mm
    \noindent Berkson, J., 1940: A note on the chi-square test, the Poisson and the Binomial. J. Amer. Statist. Assoc, 35, 362-367
    \vskip3.mm
    \noindent Biswas, B. C., N. N. Khambete and S. S. Mondal, 1989: Weekly rainfall probability analysis of dry farming tract of Tamil Nadu. Mausam, 40 (2), 197–206.
    \vskip3.mm
    \noindent Burgueo, A., Martnez, M. D., Lana, X., and Serra, C., 2005: Statistical distributions if the daily rainfall regime in Catalonia (Northeastern Spain) for the years 1950--2000. Int. J. Climatol. 25, 1381-1403. 
    %Cochran, W.G., 1954. Some methods for strengthening the common chi-square tests. Biometrics 10 (4),417451.
    \vskip3.mm
    \noindent Cochran, W. G., 1936: The chi-square distribution for the Binomial and Poisson series, with small expectations. Annals of Eugenics, 7, 207–217.
    \vskip3.mm
    \noindent Cochran, W. G., 1954: Some methods for strengthening the common chi-square tests. Biometrics. 10 (4), 417–451.
    \vskip3.mm
    \noindent Cramer, H., 1946: Mathematical methods of statistics. Princeton University Press, Princeton, NJ, 545 pp.
    \vskip3.mm
    \noindent Duan, J., Sikka, A. K., and Grant, G. E. 1995: A comparison of stochastic models for generating daily precipitation at the H. J. Andrews Experiment Forest. Northwest Science.; 69(4): 318-329.
    Fisher, R.A., 1925: Statistical methods for research workers. Hafner Publishing Company Inc., New York, 356 pp.  
    \vskip3.mm
    \noindent Fisher, R. A. and Yates, F., 1957: Statistical Tables for Biological, Agricultural and Medical Research, 5th Ed., Oliver and Boyd, Edinburgh, Scotland, 138 pp.
    \vskip3.mm
    \noindent Goel, A. K. and J. K. Singh, 1999: Incomplete gamma distribution for weekly rainfall of Unai, Himachal Pradesh. J. Agric. Eng., 36 (1), 61–74.
    \vskip3.mm
    \noindent Hargreaves, G. H., 1975: Water requirements manual for irrigated crops and rainfed agriculture. EMBRAPA and Utah State University Publication., 74-D-158, pp. 40.
    \vskip3.mm
    \noindent Kwaku, X. S., and Duke, O. 2007: Characterization and frequency analysis of one day annual maximum and two to five consecutive days maximum rainfall of Accra, Ghana. ARPN Journal of Engineering and Applied Sciences; vol. 2, no. 5: 27-31.
    \vskip3.mm
    \noindent Mooley, Diwakar A. and Crutcher, Harold L., 1968: An Application of Gamma Distribution Function to Indian Rainfall, ESSA Technical Report EDS 5, U.S. Department of Commerce, Environmental Data Service, Silver Spring, Md., Aug., 47 pp.
    \vskip3.mm
    \noindent Mooley, D. A., 1973: Gamma distribution probability model for Asian summer monsoon monthly rainfall. Monthly Weather Review, 101 (2), 160–176.
    \vskip3.mm 
    \noindent Rao, C. R. and Chakravarti, I. M., 1956: Some small sample tests of significance for a Poisson distribution. Biometrics, 12(3), 264-282.
    \vskip3.mm
    \noindent Sarker, R. P., B. C. Biswas and N. N.  Khambete, 1982: Probability analysis for short period rainfall in dry farming tract in India. Mausam, 33 (3), 269–284.
    \vskip3.mm
    \noindent Sen, Z., and Eljadid, A. G. 1999: Rainfall distribution functions for Libya and Rainfall Prediction. Hydrol. Sci. J.;4(5): 665-680.
    \vskip3.mm
    \noindent Sharma ,M. A., Singh, J. B. 2010: Use of Probability Distribution in Rainfall Analysis. New York Science Journal, 3(9), 40-49.
    \vskip3.mm
    \noindent Todorovic, P., and Woolhiser, D. A. 1975, A stochastic model of n-day precipitation. J. Appl. Meteor. ,14, 17-24.

\end{document}